\documentclass{article}

\usepackage[letterpaper, margin=1.2in]{geometry}

\usepackage{cite}
\usepackage{amsmath}
\usepackage{amssymb}
\usepackage{bm}
\usepackage{amsthm}
\usepackage{graphicx}
\usepackage{float}

\usepackage{enumitem}
\usepackage[mathscr]{euscript}
\usepackage{float}
\usepackage[version=4]{mhchem}
\usepackage{siunitx}
\usepackage{longtable,tabularx}
\usepackage{setspace}
\usepackage[normalem]{ulem}

\newtheorem{theorem}{Theorem}

\newtheorem{corollary}{Corollary}
\newtheorem{lemma}{Lemma}

\usepackage{algorithm}
\usepackage[noend]{algpseudocode}
\makeatletter
\def\BState{\State\hskip-\ALG@thistlm}
\makeatother
\algnewcommand{\Or}{\textbf{or}\,}
\algnewcommand\algorithmicswitch{\textbf{switch}}
\algnewcommand\algorithmiccase{\textbf{case}}
\algdef{SE}[SWITCH]{Switch}{EndSwitch}[1]{\algorithmicswitch\ #1\ \algorithmicdo}{\algorithmicend\ \algorithmicswitch}%
\algdef{SE}[CASE]{Case}{EndCase}[1]{\algorithmiccase\ #1}{\algorithmicend\ \algorithmiccase}%
\algdef{SE}[ELSE]{Else}{EndElse}[1]{\algorithmicelse\ #1}{\algorithmicend\ \algorithmicelse}%
\algtext*{EndSwitch}%
\algtext*{EndCase}%
\algtext*{EndElse}%

\usepackage[usenames,dvipsnames]{color}
\definecolor{darkblue}{rgb}{0,0,.6}
\definecolor{darkred}{rgb}{.7,0,0}
\definecolor{darkgreen}{rgb}{0,.6,0}
\definecolor{OliveGreen}{cmyk}{0.64,0,0.95,0.40}
\definecolor{CadetBlue}{cmyk}{0.62,0.57,0.23,0}
\definecolor{lightlightgray}{gray}{0.93}

\newcommand{\R}{\mathbb{R}}

\newcommand{\norm}[1]{\left\lVert#1\right\rVert}

\newcommand{\set}[1]{\{#1\}}

\newcommand{\argmin}{\mathop{\mathrm{argmin}}} 

\newcommand{\world}{\mathcal{W}}
\newcommand{\worlddimen}{d}
\newcommand{\worldsize}{\ell}
\newcommand{\hypercube}{H}

\newcommand{\tree}{\mathcal{T}}
\newcommand{\nodes}{\mathcal{N}}
\newcommand{\relations}{\mathcal{R}}
\newcommand{\nodea}{n}
\newcommand{\nodedepth}{k}
\newcommand{\nodelocai}{p}
\newcommand{\node}{\nodea_{\nodedepth,\nodelocai}}

\newcommand{\graph}{\mathcal{G}}
\newcommand{\vertices}{\mathcal{V}}
\newcommand{\vertex}{v}

\newcommand{\edges}{\mathcal{E}}
\newcommand{\edge}{e}
\newcommand{\edgecost}{E}

\newcommand{\agent}{\varphi}
\newcommand{\agents}{\Phi}
\newcommand{\merged}{\agents}
\newcommand{\indices}{\mathcal{I}}

\newcommand{\Path}{P}

\title{\LARGE \bf
	MAMS-A*: Multi-Agent Multi-Scale A*
}

\author{Jaein Lim$^{1}$ and Panagiotis Tsiotras$^{2}$
	\thanks{$^{1}$Jaein Lim is a graduate student at the School of Aerospace Engineering, Georgia Institute of Technology, 
		Atlanta. GA 30332-0150, USA. Email:
		jaeinlim126@gatech.edu}%
	\thanks{$^{2}$ Panagiotis Tsiotras is a Professor and David and Andrew Lewis Chair at the School of Aerospace Engineering and
		Associate Director at the Institute for Robotics and Intelligent Machines, Georgia Institute of
		Technology, Atlanta. GA 30332-0150, USA. Email: tsiotras@gatech.edu}%
}

\begin{document}
	\date{}
	\maketitle
	\thispagestyle{empty}
	\pagestyle{empty}
	
	\begin{abstract}
We present a multi-scale forward search algorithm for distributed agents to solve single-query shortest path planning problems. 
Each agent first builds a representation of its own search space of the common environment as a multi-resolution graph, it communicates with the other agents the result of its local search, and it uses received information from other agents to refine its own graph and update the local inconsistency conditions. 
As a result, all agents attain a common subgraph that includes a provably optimal path in the most informative graph available among all agents, if one exists, without necessarily communicating the entire graph. We prove the completeness and optimality of the proposed algorithm, and present numerical results supporting the advantages of the proposed approach. 
	\end{abstract}
	
	\section{Introduction}
	Humans appear to rely heavily on hierarchical structures for decision-making, especially for complex tasks (e.g., planning a route from work to home, navigating traffic, etc). Despite the fact that any plan is physically realizable as a sequence of refined actions (e.g., walking to the office door, opening the door, walking down the hallway, etc), most humans do not plan at this low level set of actions. 
	Instead, planning is achieved using a variety of abstraction levels, by aggregating the sequence of actions into high-level macro-actions (i.e., exiting the office building, changing lanes in traffic, etc), and by executing this sequence of high-level actions by refining them down~\cite{Marthi,kiebel2008hierarchy}.
	This is done because acquiring perfect knowledge about the environment is often prohibitive, and hence acting on the right level of granularity of pertinent information is imperative in order to operate in dynamic and uncertain environments. 
This observation has led many researchers to investigate multi-scale representations of the underlying search space for planning~\cite{Nissim2, Kamb, Pai2, Holte, Botea, Simon}. 

	\begin{figure}[thpb]
		\centering
		\includegraphics[width=0.7\textwidth]{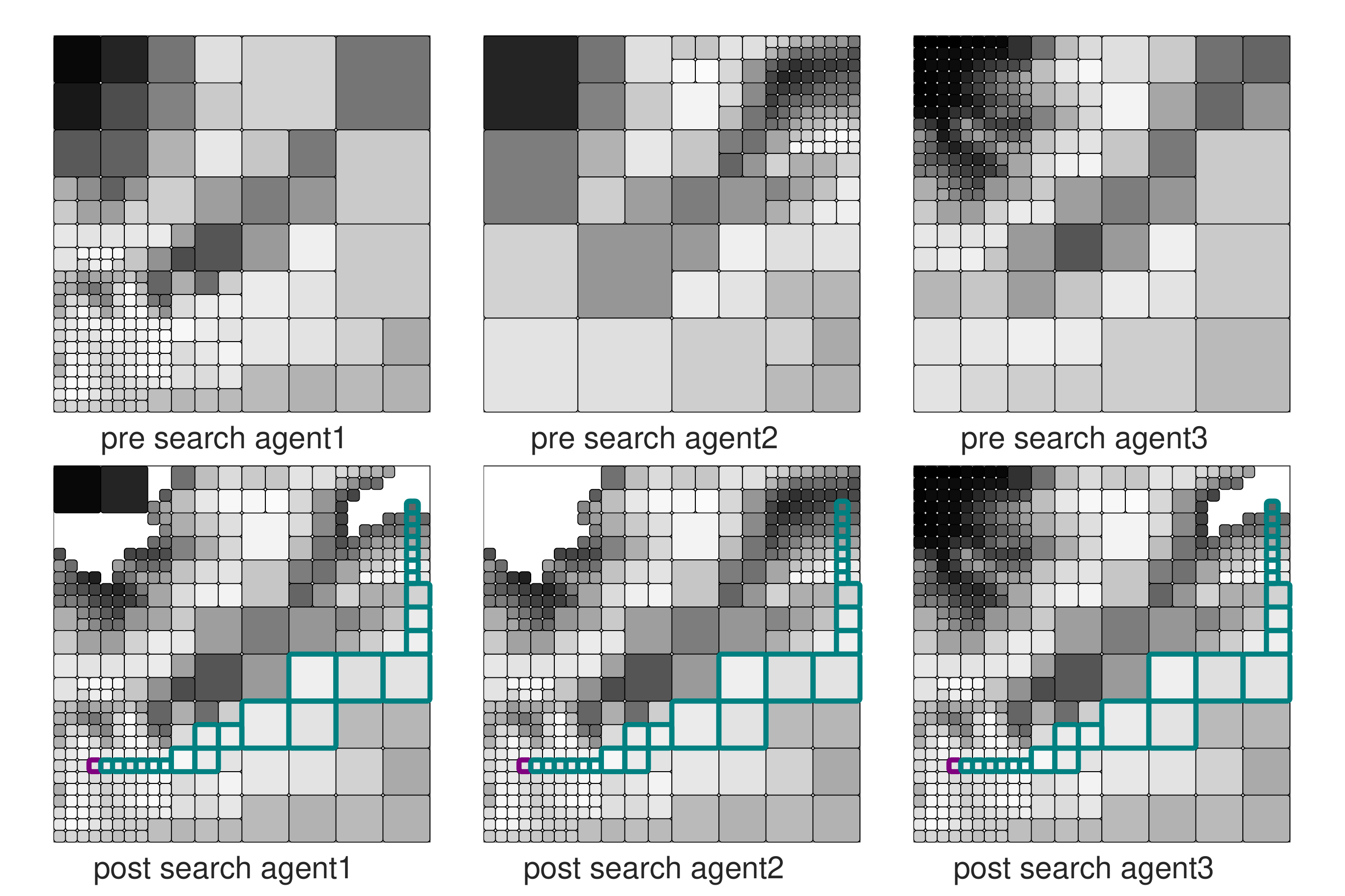}
		\caption{Top: Three agents' abstract graphs of the same environment with gray scale representing risk level. 
		Bottom: The individual agents' abstract graphs and optimal paths (green) from bottom left initial vertex (purple) to the goal vertex in the top right. 
		After the search, the agents share a common subgraph 
		that consists of the optimal path in the finest resolution abstraction available. 
		Unexpanded vertices were not communicated and are left blank.}
		\label{f:result_3a}
	\end{figure}

	Kambhampati and Davis \cite{Kamb} incorporated a top-down refinement scheme of the environment abstracted in a quad-tree for path-planning problems. 
	The obstacle region starting from the coarsest resolution level is excluded from the refinement in order to reduce the search space.
	Likewise, Pai and Reissell~\cite{Pai2} utilized top-down refinement for path-planning, specifically using the wavelet transform of a 2-D environment. 
	As the wavelet difference coefficients inform how smooth the refined path will be in the next refined level, the magnitude of the coefficients serve as an inconsistent heuristic for refinement. 
	Marthi et al~\cite{Marthi} introduced the angelic semantic for a top-down hierarchical planning, which specified for each high-level action the set of reachable states by refinement, along with the associated upper and lower bounds on the cost. 
	They showed that some pruning rules based on the upper and lower bounds on the value of the high-level action plans preserve global optimality, while the optimistic evaluation of an abstract plan serves as a heuristic for refinement~\cite{Marthi}. 
	
	More recently, the construction of multi-resolution graphs of the environment has been studied for path-planning problems to efficiently solve for a locally optimal solution~\cite{Hauer,Jung,Cowlagi, Tsiotras}. 
	These references construct fine resolution graph representations near the agent to accurately represent the environment in the vicinity of the agent, and coarser resolution representations far away from the agent to reduce the dimensionality of the search space.
	As a result, a multi-resolution framework allows a natural formulation of the perception-action loop, while still maintaining practical efficacy by utilizing a reduced search space.

	Different approaches have been studied for planning problems with multiple agents, using distributed or parallel processing techniques to increase the solution quality via inter-agent communication.
	The C-FOREST algorithm~\cite{Otte} employs multiple CPU units in parallel using the same start and goal states for randomized motion planning. Each computing agent stores its own search tree, and communicates with the other agents 
	to restrict the sampling space to focus the search so as to increase solution quality.
	Botea et al~\cite{Botea} built a hierarchical abstraction of the map, in which local roadmaps are computed in parallel with fine abstraction.
	Then, whenever a new set of start and goal locations is given, a solution is computed holistically using the coarser resolution abstraction. 
	A somewhat similar approach was used in~\cite{LHT:tac11} where a multi-scale algorithm was proposed to accelerate A*. 
	By pre-processing the environment in a hierarchical dyadic decomposition, the overall problem is divided to a nested sequence of smaller problems.
	The solutions of these smaller problems were then merged together using a ``bottom-up'' fusion algorithm to get the globally optimal path much faster than existing methods.
	Finally, Nissim and Brafman~\cite{Nissim} formulated a distributed agent planning problem
	where each agent maintains a separate search space and expands the state from its \textit{OPEN} list similar to the regular A* algorithm. 
	Each agent then informs the relevant agents (that is, those who require the current state as a precondition for their actions) of the best \textit{cost-to-come} and the heuristic estimate of that state by sending a message. 
	Once the agent receives the message from another agent, this agent adds the message to its \textit{OPEN} list only if the \textit{cost-to-come} of the state is better than its own. 
	In essence, this is an optimality-preserving pruning technique, where multiple agents communicate local information in an attempt to build their own optimal solution~\cite{Nissim}. 
	
	In this paper, we extend the distributed forward search algorithm proposed in~\cite{Nissim} to incorporate the multi-resolution framework of~\cite{Hauer} for path-planning with multiple agents. 
	Unlike other top-down refinement schemes mentioned above, which rely on the agent's own abstraction hierarchy, we refine each agent's abstract path using the information provided by the other agents.

	\section{Problem Formulation}
	
	\subsection{Multiresolution World Representation}
	
	Without loss of generality, we assume that the environment $\world \subset \R^\worlddimen$ is given as a hypercube of side length $2^\worldsize$ for some positive integer $\worldsize$. 
	The hypercube $\world$ is hierarchically abstracted as a $2^\worlddimen$-tree $\tree = (\nodes, \relations)$ using a recursive dyadic partition, where $\nodes$ is the node set and $\relations$ is the edge set describing the relations of the nodes in $\nodes$, such that each node in the tree encodes the information contained in a subset of $\world$. Specifically, each node $\node \in \nodes$ at depth $\nodedepth$ abstracts  information of the world contained in the hypercube $\hypercube(\node) \subseteq \world$ of side length $2^\nodedepth$ centered at $\nodelocai \in \world$. 
	The function $V: \nodes \to \R_+$ maps each node to some non-negative real value $V(\node)$, for example, the probability of occupancy of $\hypercube(\node)$ or a cost measure to the same region. We assume that $\tree$ is a full tree, that is, each $\node \in \nodes$ has either $2^d$ children or none. 
	The children of $\node$ are denoted by $\set{\nodea_{{\nodedepth-1}, q_i}}_{i\in [1,2^\worlddimen]} $ where $q_i = \nodelocai+2^{\nodedepth-2}e_i$ and $e_i$ is a vector in the set $\set{ [\pm1,...,\pm1]\in \R^\worlddimen$}. 
	
	\subsection{Abstract Graph Construction}

	Let $\agents =\set{\agent_i}_{i\in\indices}$ be a finite set of agents, where $\indices = \set{1,2,...,n}$ is the agent index set. 
	Each agent $\agent_i \in \agents$ builds a non-empty multi-resolution graph $\graph_{i} = (\vertices_{i}, \edges_{i})$ from $\tree$, such that $\graph_{i}$ spatially represents $\world$, where $\vertices_{i}$ is the vertex set and $\edges_{i}$ is the edge set. 
	The agent $\agent_i$ selects some nodes from $\nodes$ as its vertices in a top-down fashion, and for each node 
	$\node \in \nodes$, it selects either all of its children or none of them to maintain a dyadic representation of $\world$ using the selected vertices in $\vertices_{i}$.
	More precisely, when the agent selects the children of a node in $\nodes$, it excludes this node from $\vertices_{i}$.
	This rule ensures that the union of the regions corresponding to the vertices in $\vertices_{i}$ sufficiently covers $\world$. 
	Hence, $\vertices_{i} \subseteq \nodes$, and each $\vertex\in \vertices_{i}$ corresponds to a node of $\nodes$, but the converse is not true. 
	The node $\node$ of $\tree$ is selected as a vertex of $\graph_{i}$ if
	\begin{equation}
	\label{e:alpha}
	\norm{p-p_i}_2 - \sqrt{d} > \alpha 2^{k},
	\end{equation}
	where $p_i$ is the position of agent $\agent_i$ and $\alpha >0$ is a user-specified parameter.
	
	Similarly to the nodes, we denote the hypercube covered by vertex $\vertex\in \vertices_{i}$ as $\hypercube(\vertex)$, 
	such that if vertex $v\in\vertices_i$ corresponds to node $n\in \nodes$, then $\hypercube(v)=\hypercube(n)$.
Let  a set of vertices $W \subseteq \mathcal{N}$.
	We define $H(W)$ be the hypercube covered by $W$, that is, $H(W)= \bigcup_{v\in W} H(v)$.
	Two vertices $\vertex$ and $\vertex'$ in $\vertices_i$ are neighbors if the union of the boundaries of the corresponding hypercubes is neither empty nor a singleton. 
	Each edge $\edge \in \edges_{i}$ assigns a non-negative real value $\edgecost(\vertex,\vertex')$ to a pair of vertices $(\vertex,\vertex')\in \vertices_{i}\times \vertices_{i}$ if and only if the corresponding vertices are neighbors. 
	The edge value $\edgecost(\vertex,\vertex')$ is the cost to traverse from $\vertex$ to $\vertex'$, and is defined as~\cite{Hauer}
	\begin{equation}
	\label{e:cost}
	\edgecost(\vertex,\vertex') = 2^{dk}(\lambda_1 V(\node) + \lambda_2),
	\end{equation}
	where $\lambda_1, \lambda_2 \in (0,1]$ are weights used to penalize the content and the abstraction level, respectively, of the corresponding node $\node$ of $v'$. 
	
	\subsection{Merged Graph}
	
	Consider an arbitrary vertex set $W$ 
	which includes the vertices $u$ and $v$. 
	A vertex $u$ is defined to be a child vertex of $v$ in $W$, if $\hypercube(u) \subset \hypercube(v).$
	We define a vertex $w \in W$ to be a \textit{fine} vertex in $W$ if there are no children vertices of $w$ in $W$. 
	Otherwise, $w$ is defined to be a \textit{coarse} vertex in $W$.
	Hence, every vertex in the individual agent's graph $\graph_{i}=(\vertices_{i}, \edges_{i})$ is finest by construction, as the agent excludes the parent node from $\vertices_{i}$ if the agent chooses the children nodes to include in $\vertices_{i}$.
	
    For the set of agents $\agents$, we define the merged graph $\graph_\merged =(\vertices_\merged, \edges_\merged)$ 
    with vertices consisting of only the fine vertices of the union $\vertices_\indices = \bigcup_{i\in \indices}\vertices_{i}$ of the vertices of all agents' graphs $\graph_i$, 
    and edge set $\edges_\merged$ that assigns the edge cost $E(\vertex,\vertex')$ to every pair of neighboring $(\vertex, \vertex')\in\vertices_\merged \times \vertices_\merged$. 
    The merged graph excludes a vertex $v$ from $ \vertices_\indices$ if there exist children vertices of $v$ that cover the region covered by $v$. Note that $\vertices_\merged \subseteq \vertices_\indices$ is sufficient to spatially represent $\world$ compactly.
    Figure~\ref{f:tree_graph} illustrates the construction of the graph $\graph_\merged$.

	\begin{figure}[thpb]
		\centering
		\includegraphics[width=0.8\textwidth]{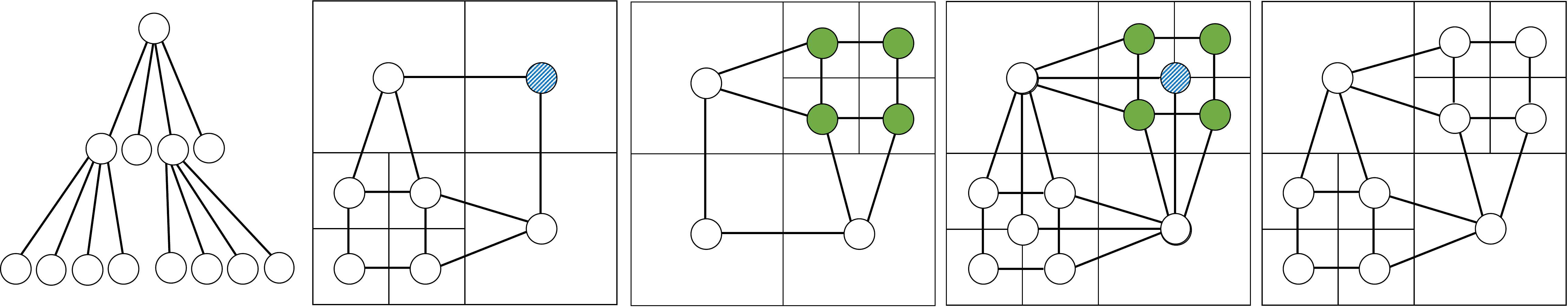}
		\caption{Example of problem formulation for $\world \subset \R^2$, from left to right: a) $\tree$ - abstracted world encoded in a quad-tree; 
		b) $\graph_{1}$ - agent $\agent_1$'s graph construction from $\tree$; 
		c) $\graph_{2}$ - agent $\agent_2$'s graph construction from $\tree$; 
		d) $\vertices_{\indices =\set{1,2}}$ - hashed blue vertex of $\vertices_{1}$ is coarse and solid green vertices of $\vertices_{2}$ are fine in $\vertices_{1} \cup \vertices_{2} $;
	    and e) $\graph_\merged$ - merged graph from both $\graph_{i}$ and $\graph_{j}$}
		\label{f:tree_graph}
	\end{figure}

	\subsection{Path-Planning Problem}     
	
	A path $\Path =(v_0, v_1, \ldots , v_k)$ on the graph $\graph = (\vertices, \edges)$ is an ordered set of vertices $v_i \in \vertices$, $i=0,\ldots,k$ 
	such that for any two consecutive vertices there exists an edge $\edge \in \edges$. 
	Given $\graph$, an initial vertex $\vertex_{\mathrm{init}} \in \vertices$ and a goal vertex $\vertex_{\mathrm{goal}} \in \vertices$, we define the path-planning problem $\Pi$, as the problem to find a path $\Path=(\vertex_\mathrm{init}, \ldots , \vertex_\mathrm{goal})$ in $\graph$. 
	An optimal path $P_\graph (\Pi)$ or $P_\graph$ for a problem instance $\Pi = \Pi(\vertex_\mathrm{init},\vertex_\mathrm{goal})$ with initial and goal vertices $\vertex_\mathrm{init}$ to $\vertex_\mathrm{goal}$ is a path having the smallest cost over all paths from $\vertex_\mathrm{init}$ to $\vertex_\mathrm{goal}$. To simplify the notation, we use $P_\merged$ instead of $P_{\graph_\merged}$ to denote an optimal path in the merged graph $\graph_\merged$.

	We denote the~\textit{cost-to-come} with $g:\vertices \to \R_+$ and the heuristic \textit{cost-to-go} with $h:\vertices \to \R_+$, which both assign each vertex in $\vertices$ a non-negative real value. 
	For each individual agent, these functions may assign different values to the same vertex. 
	Hence, we use a subscript, namely $g_{i}(\vertex_k)$, to denote the cost accumulated from the initial vertex $\vertex_\mathrm{init} = \vertex_0$ to 
	$\vertex_k$ in the agent $\agent_i$'s graph, and use $h_{i}(\vertex_k)$ as the heuristic estimate from $\vertex_k$ to the goal vertex $\vertex_\mathrm{goal} = \vertex_g$ in the agent $\agent_i$'s graph representation. 
	For each agent, we define the evaluation function $f_{i}(\vertex_k)= g_{i}(\vertex_k) + h_{i}(\vertex_k)$. 
	For the single agent case with uniform resolution, the problem formulation reduces to the well-known shortest-path-planning problem on a graph.    
	
	\section{The Multi-Agent Multi-Scale A* (MAMS-A*) Algorithm}
	
	Some agents may have a finer resolution abstraction in a region than others, and we are interested in finding an optimal path that would be optimal in the finest resolution abstraction of the world if agreed upon by all agents. 
	One na\"ive way to achieve this objective would be to build a single merged graph abstraction $\graph_{\agents}$ from all agents in $\agents$ first, and then search for an optimal path $P_\merged$ within this new graph. 
	Consequently, all agents will perform a search in the most informative graph available among them. 
	However, as much this approach refines the abstraction level of the graph, it also enlarges the search space. 
	In addition, the amount of information communicated among the agents is increased unnecessarily by broadcasting the entire graph, whether or not the received information can improve the solution quality. 
	Thus, construction of such a graph dilutes the benefit of using multi-resolution abstractions during the search. 
	
	To alleviate the mentioned issues, we propose an algorithm in which each agent  $\agent_i$ broadcasts a message containing the vertex information (i.e., $\langle s, g_{i}(s), h_{i}(s) \rangle $) only when 
	$\agent_i$ expands its own vertex $s$. 
	Also, each agent processes the received message only if it contains finer resolution information about the environment or if the messaged vertex has not been expanded or it has a better \textit{cost-to-come} than the agent's own cost-to-come value. 
	Hence, we limit communication only to the expanded vertices. 
	If the heuristic used by the agents is consistent, then the algorithm expands and broadcasts only the vertices with the lowest possible \textit{total cost}.
	In addition, we assume that every broadcasted message reaches all agents in $\agents$. 
	We present the pseudo-algorithm below.

	\begin{algorithm}
		\caption{MAMS-A* for $\agent_i$}\label{a:search}
		\begin{algorithmic}[1]
			\While {$\exists$ active agent $\in \agents$} \label{ap:terminate_cond}
			\State ProcessMessage
			\If {$\agent_i$ is active}
			\State $s \gets$ pop front \textit{OPEN}
			\State Expand($s$)
			\State Publish($s$)
			\EndIf
			\EndWhile
		\end{algorithmic}
	\end{algorithm}

	\begin{algorithm}
		\caption{ProcessMessage}\label{a:processmsg}
		\begin{algorithmic}[1]
			\For{message = ($\langle s, g_{j}(s), h_{j}(s) \rangle $) $\in$ message\_que}
			\Switch {message}
			\Case {$ \exists $ $v$ $\in$ $\vertices_i$ such that $\hypercube(v)=\hypercube(s)$} \label{al:inconsistencyupdate}
			\If {$v$ not expanded \Or $g_{i}(v) > g_{j}(s)$ }
			\State Adopt($s$)
			\EndIf
			\EndCase
			\Case {$ \exists $ $v$ $\in$ $\vertices_i$ such that $\hypercube(s) \subset \hypercube(v)$} \label{al:updategraph}
			\State add $s$ to $\vertices_i$
			\For{$v'$ $\in$ neighbors($v$)}
			\State remove $e(v,v')$ from $\edges_i$ \label{al:updategraph-remove}
			\State add $e(s,v')$ to $\edges_i$ if exists
			\EndFor
			\State remove $v$ from $\vertices_i$
			\State remove $v$ from \textit{OPEN} \Or \textit{CLOSED}
			\State Adopt($s$)
			\EndCase
			\Case {$ \exists$ $v$ $\in \vertices_i$ such that $\hypercube(v) \subset \hypercube(s)$}
			\State discard message
			\EndCase
			\Else {} \label{al:localupdate}
			\State add $s$ to $\vertices_i$ \label{al:local addition}
			\For{$v'$ $\in \vertices_i$}
			\State add $e(s,v')$ to $\edges_i$ if exists
			\EndFor
			\State Adopt($s$)
			\EndElse
			\EndSwitch
			\EndFor
		\end{algorithmic}
	\end{algorithm}

	\begin{algorithm}
		\caption{Expand($s$)}\label{a:expand}
		\begin{algorithmic}[1]
			\If {s is a goal vertex \Or \textit{OPEN} = $\emptyset$ } \label{al:inactivate}
			\State inactivate $\agent_i$ \label{ap:inactivate}
			\EndIf
			\Else
			\State move s to \textit{CLOSED}
			\For {$s' \in$ neighbors(s)}
			\If {$s'\notin$ \textit{CLOSED} }
			\If {$g_{i}(s') > g_{i}(s)+ \edgecost(s,s')$}
			\State $g_{i}(s') \gets g_{i}(s)+ \edgecost(s,s')$ \label{al:cost_propagate}
			\State $f_{i}(s') \gets g_{i}(s') + h_{i}(s') $
			\State add $s'$ to \textit{OPEN}
			\State predecessor($s'$) $\gets$ s
			\EndIf
			\EndIf
			\EndFor
			\EndElse
		\end{algorithmic}
	\end{algorithm}

	\begin{algorithm}
		\caption{ Adopt($s$) }\label{a:adopt}
		\begin{algorithmic}[1]
			\If{$\exists$ predecessor(s) in $\vertices_i$} \label{al:predecessorcheck}
			\State $g_{i}(s) \gets g_{j}(s)$ \label{ap:msg update}
			\State $h_{i}(s) \gets \max(h_{i}(s) ,h_{j}(s))$
			\State $f_{i}(s) \gets g_{i}(s) +h_{i}(s) $
			\State add $s$ to \textit{OPEN}
			\State reactivate $\agent_i$ \label{ap:reactivate}
			\EndIf
		\end{algorithmic}
	\end{algorithm}   
		
	The proposed multi-agent multi-scale A* algorithm comprises of three major procedures: 1) \textsf{ProcessMessage}, 2) \textsf{Expand}, and 3) \textsf{Publish}, and is summarized in Algorithm~\ref{a:search}.
	Each individual agent $\agent_i\in \agents$ repeats these procedures until the finest goal vertex is expanded and processed among the agents in $\agents$. 
	Once agent $\agent_i$ expands the goal vertex of its graph $\graph_i$, this agent is inactivated by Line~\ref{ap:inactivate} of Algorithm~\ref{a:expand}. Inactive agents do not expand any vertices nor publish any messages until a message from another agent reactivates them by Line~\ref{ap:reactivate} of Algorithm~\ref{a:adopt}. 
	The algorithm starts with all agents being active, and terminates when all agents become inactive (Line~\ref{ap:terminate_cond} of Algorithm~\ref{a:expand}). 
	Suppose $u_i$ and $u_j$ are the goal vertices of agents $\agent_i$ and $\agent_j$ respectively, such that $\hypercube(u_j) \subset \hypercube(u_i)$. 
	If agent $\agent_i$ expands its goal vertex $u_i\in \vertices_i$ (which is coarse in the vertex set $\vertices_\indices$ since $\hypercube(u_j) \subset \hypercube(u_i)$)
	 before the vertex $u_j\in \vertices_j$ 
	 is broadcasted by another agent $\agent_j$, then $\agent_i$ will not expand any vertices nor publish any messages until this agent is reactivated by an incoming message containing the vertex $u_j$ (Line~\ref{ap:reactivate} of Algorithm~\ref{a:adopt}). 
	On the other hand, if a message containing vertex $u_j$ is broadcasted to agent $\agent_i$ before $\agent_i$ expands its goal vertex $u_i$, then $u_i$ will be simply removed from the agent $\agent_i$'s graph. 
	Hence, every agent will expand the fine goal vertex in $\vertices_\merged$ at least once before termination.

	In case an agent receives finer resolution information about one of its vertices (Line~\ref{al:updategraph} of Algorithm~\ref{a:processmsg}), the agent removes this vertex from its graph and from its \textit{OPEN} or \textit{CLOSED} lists.
	When agent $\agent_i$ receives the expansion result of the vertex from agent $\agent_j$ that already exists in its graph $\graph_i$, Line~\ref{al:inconsistencyupdate} of Algorithm~\ref{a:processmsg} puts the expansion result in the \textit{OPEN} list only if the vertex has not been expanded or if it has better \textit{cost-to-come} value. 	
	When a new vertex is received that has no related vertex in $\graph_i$ because the coarse vertex has already been removed, then Line~\ref{al:localupdate} of Algorithm~\ref{a:processmsg} simply updates the local graph and puts the expansion result in the \textit{OPEN} list.

	\begin{figure}[thpb]
		\centering
		\includegraphics[width=0.8\textwidth]{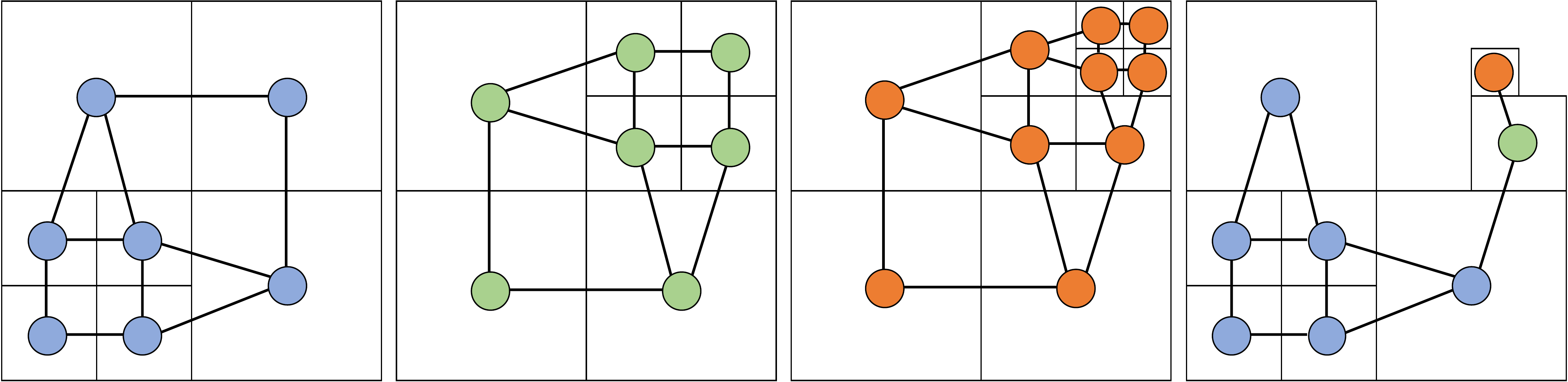}
		\caption{Illustrative example of Algorithm~\ref{a:processmsg}, from left to right: a) $\graph_{i}$ - agent $\agent_i$'s original graph; 
		b) $\graph_{j}$ - agent $\agent_j$'s graph; 
		 c) $\graph_{k}$ - agent $\agent_k$'s graph; 
		 d) agent $\agent_i$'s modified graph after receiving messages from $\agent_j$ and $\agent_k$}
		\label{f:cuckoo_ex}
	\end{figure}
	
	By removing a vertex from its graph upon receiving fine resolution information about this vertex, we are preemptively preventing each agent to construct a path that includes coarse resolution vertices from $\vertices_\indices$.
	Instead, if necessary, the agent will receive fine vertices passing through the region originally covered by the removed vertex, as the other agents with fine resolution information available expand and broadcast (see Figure~\ref{f:cuckoo_ex}, for instance). 
	Indeed, if an optimal path $P_\merged$ passes through the region covered by the removed vertices, then the corresponding segment of $P_\merged$ will be expanded and broadcasted by other agents. 
	Hence, path connectivity for the agent will be restored despite of the removal of coarse vertices. 
	We give the formal proof of optimality and completeness of the MAMS-A* algorithm in the next section.

	\section{Analysis}
	
	We prove the completeness of the MAMS-A*
	algorithm and the optimality of the solution with respect to the merged graph by extending a well known result of A* \cite{Hart}. 
	We will use Lemma~\ref{lemma: baby lemma} and its corollary to show the completeness of the algorithm regardless of the removal of coarse vertices, assuming that every broadcasted message arrives to all the agents.

	\begin{lemma}
		\label{lemma: baby lemma}
	 Let $\vertices_\indices$ be the collection of all agents' vertices, that is, $\vertices_\indices = \bigcup_{i\in \indices}\vertices_i.$ 
	 For any coarse vertex $u$ in $\vertices_\indices$ of an agent $\agent_i$, there exists a set of the fine vertices $W \subseteq \vertices_\indices$ such that $H(W) = H(u)$.  
	\end{lemma}
	
	\begin{proof}
		Let $k$ be some integer such that  vertex $v_k$ has corresponding region $H(v_k)$ with side length $2^k$.  
		Suppose vertex $v_{k}$ of agent $\agent_i$ is coarse in $\vertices_\indices$, that is, there exists at least one vertex $v_{m}\in \vertices_\indices$ such that $H(v_m)\subset H(v_k)$ for $m<k$. 
		Let $M =\set{m: H(v_m)\subset H(v_k), v_m\in \vertices_\indices}$. 
		Then $M$ is non-empty, closed and bounded, since $v_k$ is coarse and there is only a finite number of vertices in $\vertices_\indices$. 
		Let $n$ be the minimum of $M$. 
		Then $v_{n}$ is a fine vertex corresponding to $v_k$, that is, $\hypercube(v_n) \subset \hypercube(v_k)$ with $n\leq m <k$. Without loss of generality, let $\vertices_j$ be the vertex set of agent $\agent_j$ which includes $v_n$.
		Since $\hypercube(\vertices_j) =\world$, and $\vertices_j$ is the dyadic partition of $\world$, there exists a subset $V_{n}\subseteq \vertices_j$ of cardinality $2^\worlddimen$ containing $v_{n}$ whose elements have the same corresponding parent node, namely, $\mu_{n+1}$, and $H(V_{n})=H(\mu_{n+1})$.
		If $n=k-1$, then $H(V_{k-1})=H(v_{k})$ and $V_{k-1} \subseteq \vertices_j \subseteq \vertices_\indices$. Since there cannot exist a finer vertex than $v_n$, the claim holds. 
		
		Consider now $n< k-1$. 
		Since the agent $\agent_j$ has the fine  resolution vertex $v_{n}$, the agent had once selected the node 
		 $\mu_{n+1}\in \nodes$, the parent node of $\mu_{n}\in \nodes$, where $H(\mu_n) = H(v_n)$, from tree $\tree$ as its vertex during the construction of its abstract graph $\graph_j$. 
		 Hence, there exists a set $W_1 \subseteq \vertices_{j}$ which contains either siblings of $v_{n+1}$ or their children, such that $H(W_1) = H(v_{n+2})$. In the same way, since agent $\agent_j$ had once selected $\mu_{n+1}\in \nodes$, it also had selected the parent node $\mu_{n+2} \in \nodes$ 
		 during the construction of its abstract graph $\graph_j$.
		 Hence, there exists a set $W_2 \subseteq \vertices_{j}$ which contains either siblings of $v_{n+2}$ or their children, such that $H(W_2) = H(v_{n+3})$. Repeating the same argument $k-n-1$ times, we have a set $W_{k-n-1}\subseteq \vertices_{j}$ such that $H(W_{k-n-1})=H(v_{k})$, and $W_{k-n-1}$ has vertices with side length not greater than level $k-1$. Repeating the entire argument for all coarse vertices of $W_{k-n-1}$ completes the proof.
	\end{proof}
	
	\begin{figure}[thpb]
		\centering
		\includegraphics[width=0.8\textwidth]{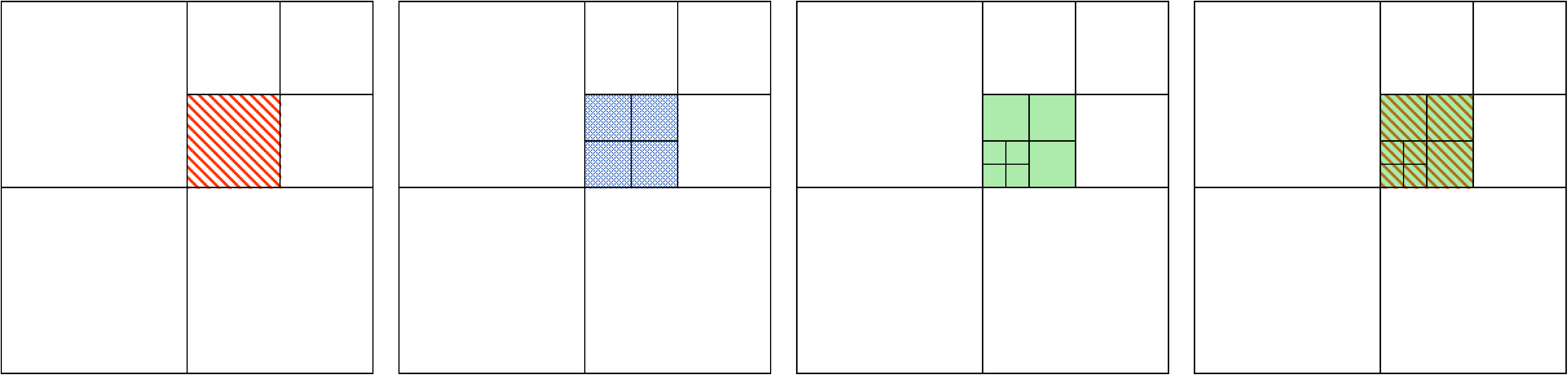}
		\caption{Illustrative example Lemma~\ref{lemma: baby lemma}, from left to right:
			a) agent $\agent_1$'s vertex set $\vertices_{1}$; 
			b) agent $\agent_2$'s vertex set $\vertices_{2}$; 
			c) agent $\agent_3$'s vertex set $\vertices_{3}$; 
			d) $\vertices_{\indices=\set{1,2,3}}$ - $\agent_1$'s hashed red vertex is coarse in $\vertices_{\indices=\set{1,2,3}}$, and there exists a set of the fine vertices (green solid) covering the same region.   
			}
		\label{f:lemma_pic}
	\end{figure}
	
	\begin{corollary}  		\label{coro: baby corollaray}
		For any coarse vertex $u$ of an agent $\agent_i$, there exists a set of fine vertices $W \subseteq \vertices_\merged$ such that $H(W) = H(u)$.
	\end{corollary}
	
	\begin{proof}
		By Lemma~\ref{lemma: baby lemma}, for any coarse vertex $u\in \vertices_\indices$, there exists a set of the fine vertices $W\subseteq \vertices_\indices$ and $H(W)=H(u)$. 
		Since $\vertices_\merged \subseteq \vertices_\indices$ is the largest finest resolution vertex
		set in $\vertices_\indices$, and $W$ consists only the fine vertices, it follows that $W\subseteq \vertices_\merged$.    
	\end{proof}
	
	The following lemma and the corollary will be used to prove the optimality of the path with respect to the merged graph $\graph_\merged$.

	\begin{lemma} 		\label{lemma: the lemma}
		Let $v$ be a vertex that has not been expanded by any agent. 
		For any optimal path $P_\merged$ from $s$ to $v$, there exists an agent $\agent_i \in \agents$ that has either an open vertex $v'$ or has an incoming message $\langle v', g_{j}(v'), h_{j}(v') \rangle $ from $\agent_j$, $i\neq j$
		such that $v'$ is on $P_\merged$ and $g_i (v') = g_\merged(v')$. 
	\end{lemma}
	
	\begin{proof}
		Let $P_\merged = (s=v_0, v_1, ..., v_k = v)$. If $v'=s$, then the lemma is trivially true, since $g_{i} (s) = g_\merged(s) = 0$. 
		Suppose now that $s$ is closed by all agents, and let $\Delta$ be the set of closed vertices by some agents in $P_\merged$,  such that $g_i(\delta) = g_\merged(\delta)$ for all $\delta \in \Delta$. 
		Then $\Delta$ is not empty, since $s\in \Delta$. Let $v^*$ be the element of $\Delta$ with the highest index closed by agent $\agent_i$, that is $v^* =\argmin_{\delta\in\Delta} f_i(\delta)$.
		Clearly, $v^* \neq v$, as $v$ is not closed. 
		Let $v'$ be the successor of $v^*$ on $P_\merged$ (possibly $v' = v$). 
		If $v'$ is reachable by $\agent_i$, then $g_i(v') = g_i(v^*) + \edgecost(v^*, v') = g_\merged(v')$, because $v'$ is on $P_\merged$. 
		Otherwise, there exists $\agent_j$ who can reach $v'$ via processing the message $\langle v^*, g_{i}(v^*), h_{i}(v^*) \rangle $, because no vertices in $P_\merged$ can be removed. 
		By Line~$\ref{al:cost_propagate}$ of Algorithm~\ref{a:expand},  $g_j(v') = g_{j}(v^*) + \edgecost(v^*, v') =   g_\merged(v')$. 
		Then, the message $\langle v', g_{j}(v'), h_{j}(v')\rangle$ is broadcasted, and eventually $g_{j}(v') = g_i(v') =  g_\merged (v')$ by Line~$\ref{ap:msg update}$ of Algorithm~\ref{a:adopt}. 
		Hence the claim holds. 
	\end{proof}
	
	\begin{corollary}  		\label{coro: the corollaray}
		Suppose $h_k$ is admissible for all $k\in \indices$, that is, $h_k \le h_\merged$, where $h_\merged$ is the true \textit{cost-to-go} in $\graph_\merged$, and suppose the algorithm has not terminated. 
		Then, for any optimal path $P_\merged$ from the initial vertex $s$ to any goal vertex, there exists an agent $\agent_i$ which either has an open vertex $v'$ or has an incoming message containing $v'$, such that $v'$ is on $P_\merged$ and $f_{i}(v') \leq f_\merged(s)$. 
	\end{corollary}
	
	\begin{proof}
		By Lemma~\ref{lemma: the lemma}, there exists an agent $\agent_i$ which either has an open vertex $v'$ or has an incoming message containing $v'$ on $P_\merged$ with $g_i (v') = g_\merged(v')$. Then
		\begin{equation*}
		\begin{aligned}
		f_i (v') &= g_i (v') + h_i (v') \\
		& =  g_\merged (v') + h_i (v') \\ 
		& \leq g_\merged (v') + h_\merged (v') = f_\merged (v').
		\end{aligned}
		\end{equation*} Since $v'$ is on the optimal path $P_\merged$, $f_\merged (v') = f_\merged (s)$, which completes the proof. 
	\end{proof}

	\begin{theorem}
		With admissible heuristisc $h_i(s)$, $i\in \indices$, MAMS-A* terminates in a finite number of iterations by finding an optimal solution, if one exists, in the merged graph $\graph_\merged$.
	\end{theorem}
	
	\begin{proof}
		We prove this theorem by contradiction. Suppose the algorithm does not terminate by finding an optimal path to a goal vertex in the merged graph $\graph_\merged$. There are three cases to consider:
		\begin{enumerate}[leftmargin=*]
			\item \textit{The algorithm terminates at a non-goal. } This contradicts the termination condition (Line~\ref{ap:terminate_cond} of Algorithm \ref{a:search}) since the agents become inactive only if they expand a goal vertex or the $\textit{OPEN}$ list is empty (Line~\ref{al:inactivate} of Algorithm \ref{a:expand}). At least one agent has non-empty $\textit{OPEN}$ if a goal vertex has not been expanded.
			
			\item \textit{The algorithm fails to terminate.} Since there is a finite number of non-goal vertices, a finite number of agents, and a finite number of non-cyclic paths from the start vertex $s$ to any vertex $v$ with non-negative edge cost, a vertex will be closed forever by all agents or removed permanently from the search space by Line~ \ref{al:updategraph} of Algorithm \ref{a:processmsg}. Hence, the only possibility left for an agent to remain active without reaching the goal vertex is when the graph is disconnected along the path from the initial vertex to the goal vertex by Line~\ref{al:updategraph-remove} of Algorithm \ref{a:processmsg}. 	
			Suppose a coarse vertex $u$ of $\graph_{i}$ was removed by Line~\ref{al:updategraph-remove} of Algorithm~\ref{a:processmsg}.  
			By Corollary~$\ref{coro: baby corollaray}$, there exists a set $W$ of fine vertices such that $H(W) = H(u)$. 
			Hence, if $H(u)$ contains the part of the path to the goal, the vertices of $W$ will be added to graph $\graph_{i}$ 
			by Line~\ref{al:local addition} of Algorithm~\ref{a:processmsg}. 
			Moreover, any vertex $w$ in $W$ cannot be removed because it is a fine vertex. 
			Hence,  a path to the goal vertex that may have become 
			disconnected by the removal of the coarse vertex $u$ will be eventually restored by the set of the fine vertices in $W$.  
			This contradicts the assumption that the algorithm failed to terminate.

			\item \textit{The algorithm terminates at a fine resolution goal without achieving the minimum cost in the merged abstraction $\merged$.} 
			Suppose the algorithm terminates at some goal vertex $v$ with 
			$f_j (v) > f_\merged (v)$. 
			By Corollary \ref{coro: the corollaray}, just before termination, there existed an agent $\agent_i$ which had an open node $v'$, or had an incoming message containing $v'$, such that $v'$ is on an optimal path and $f_i(v') \leq f_\merged(s)$. Thus, at this stage, $v'$ would have been selected for expansion rather than $v$, or at least one agent would have been reactivated by the message containing $v'$, contradicting the assumption the algorithm terminated without achieving the minimum cost.
		\end{enumerate}
	\end{proof}

	\section{Discussion}  

	Note that since the attained solution does not necessarily imply resolution completeness, it may be prudent to alleviate the computational burden of message broadcasting by sequentially sending messages from one agent to another at the expense of optimality.
	Let the sequence of agents be $(\agent_1, \agent_2, \ldots, \agent_n)$, such that $\agent_1$ passes its expansion result to $\agent_2$, and $\agent_2$ passes it to $\agent_3$, and so forth. 
	The agents will not loose any information that has been deemed optimal by the previous agents, and therefore the next agent $\agent_{i+1}$ will always make decisions based on augmented information provided by the previous agent $\agent_i$, resulting in an optimal path in the partially merged graph. Nonetheless, the last agent $\agent_n$'s path may not be necessarily the same as $\Path_\merged$, since a segment of $\Path_\merged$ could be ignored by an agent $\agent_i$ without the information of some agent $\agent_{j}$ for $j>i$, and thus this segment may not be passed onto the next agent $\agent_{i+1}$. 
	However, in our experiments the results were shown to be very close to the optimal one.    
	
	We also adopted the backtracking algorithm presented in \cite{Hauer} to incorporate a further refinement scheme for resolution completeness, i.e., to find a feasible path in the fine resolution space. 
	At each iteration, the agent traverses along the abstract solution path and stores only the fine resolution path segment to its memory, then re-solves the shortest path problem at a new vertex with different abstraction until the goal vertex is reached.
	If no solution can be found from the current vertex $v$, the agent backtracks to the previous vertex $u$ in the accumulated path and removes the edge $\edge = (u,v)$ from the graph to avoid a cyclic iteration. 
	The proof of completeness is omitted for brevity, and instead we refer the reader to \cite{Hauer}.
	Note that it is desirable to backtrack as early as possible, however this depends on the environment and the choice of $\alpha$. 
	A larger value of $\alpha$ makes the algorithm easier to backtrack earlier, as the agent attains more information far away from the current location than with a smaller value of $\alpha$. 
	However, the cardinality of the search space also increases with increasing $\alpha$, and therefore the right granularity of the abstract graph for optimal performance is not known a priori.   
	
	\section{Numerical Results}

	The MAMS-A* algorithm was implemented within the Robot Operating System (ROS) framework for modularized agents. 
	In a single workstation, multiple ROS nodes were generated to solve a single query planning problem cooperatively. 
	The communication among the nodes was made via ROS messaging. 
	
	Different number of agents were generated at different locations including the start and goal positions in a 2D labeled map; see Figure~\ref{f:result_3a}.
	In our experiments, we chose $\lambda_1$ and $\lambda_2$ to be 0.999 and 0.001, respectively, and we set $V(n_{k,p})\in [0,1]$ to be the risk level averaged over the region $\hypercube(n_{k,p})$.
	Different values of $\alpha$ were chosen with varying size of the search space represented as depth, e.g., depth 7 corresponds to $2^{7d}$ nodes. 
	Each result was then normalized by the resolution complete solution computed with a regular A* in the corresponding depths. 
	Figure~\ref{f:instance_result} shows the result after one iteration.  
	\begin{figure}[thpb]
		\centering
		\includegraphics[width=0.46\textwidth]{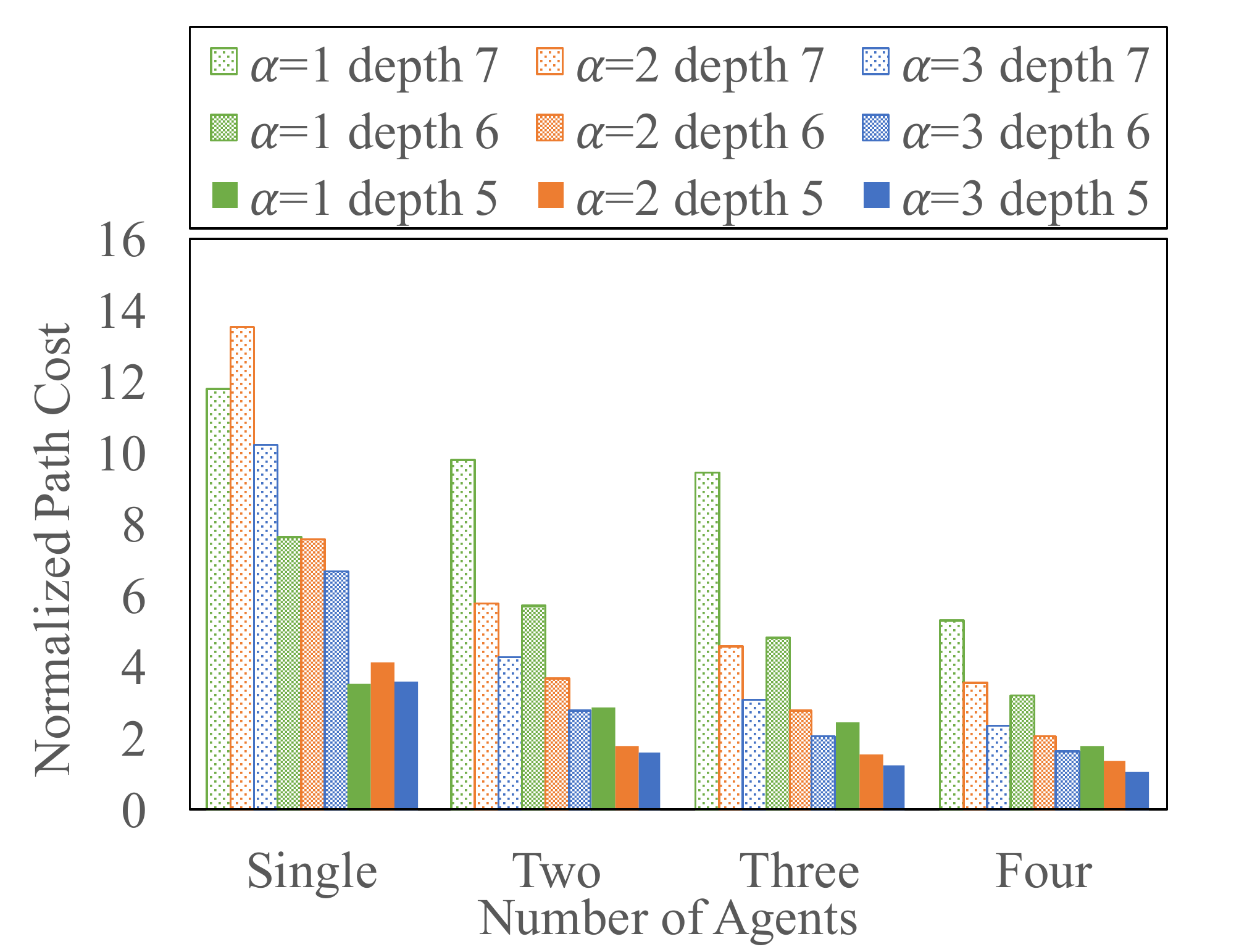}
		\includegraphics[width=0.46\textwidth]{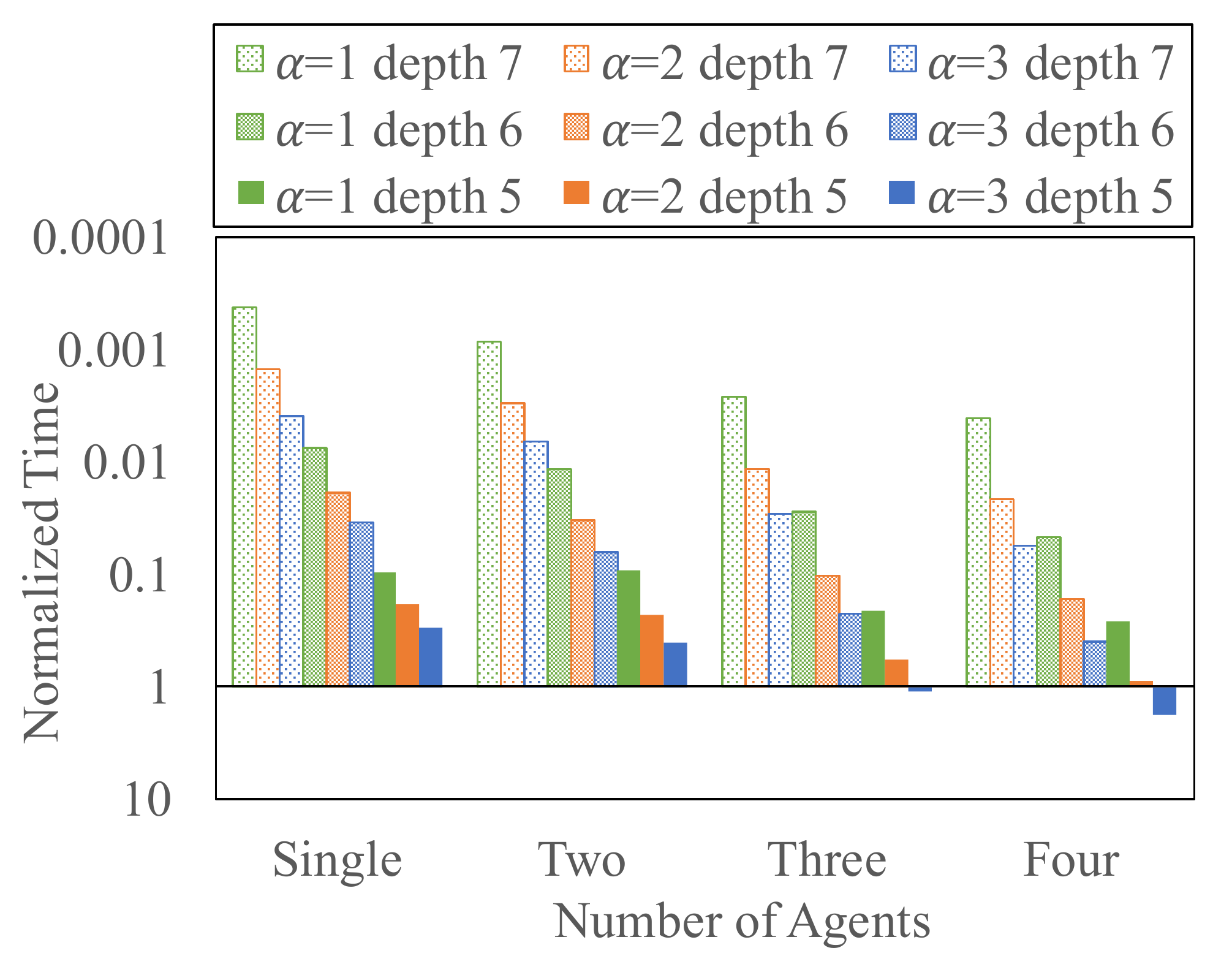}
		\caption{Computation time and solution cost of abstract path for different number of agents and parameters normalized by regular A* result. Left: normalized path cost. Right: normalized time.}
		\label{f:instance_result}
	\end{figure}
	
	The computational advantage of the proposed algorithm is most prominent when the original search space is large and $\alpha$ is small (e.g., $\alpha =1$ depth 7), as the abstraction reduces the search space most significantly. 
	The algorithm finds an abstract path three orders of magnitude faster compared to the regular A*. 
	As we penalize the abstraction in the cost function defined in equation ($\ref{e:cost}$), the cost of the abstracted path is substantially worse compared to the resolution optimal solution especially for the single agent case. 
	The increased number of agents improves the solution quality of the abstract path.  
	
	A similar comparison was made and plotted in Figure~\ref{f:holistic_result}, but using instead the backtracking algorithm to solve for a fine resolution path. 
	At each iteration, only one agent was allowed to move, while the other agents remained stationary. This was done to demonstrate the application of the algorithm for cooperative agents with different goal locations. 
	
	\begin{figure}[thpb]
		\centering
		\includegraphics[width=0.46\textwidth]{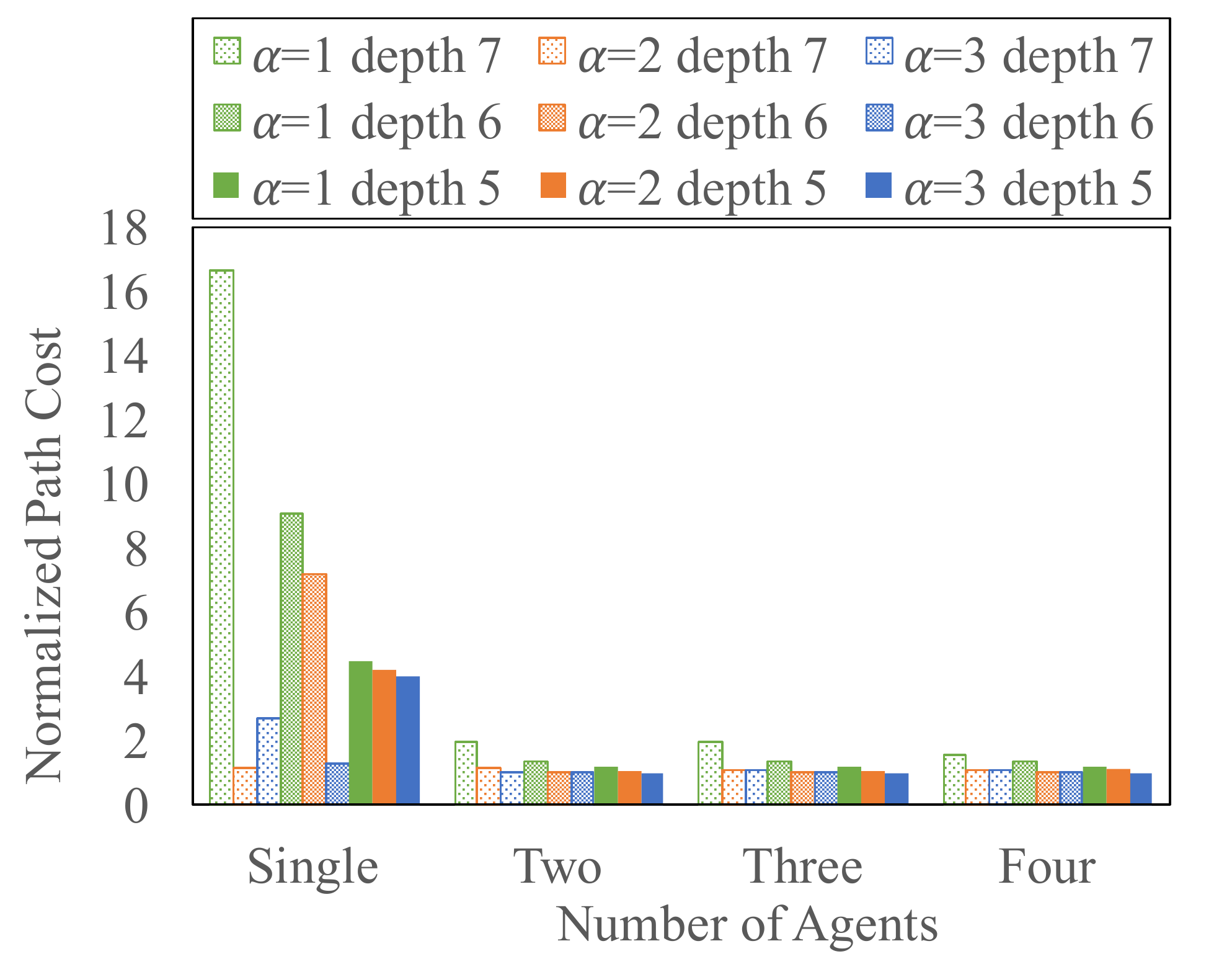}
		\includegraphics[width=0.46\textwidth]{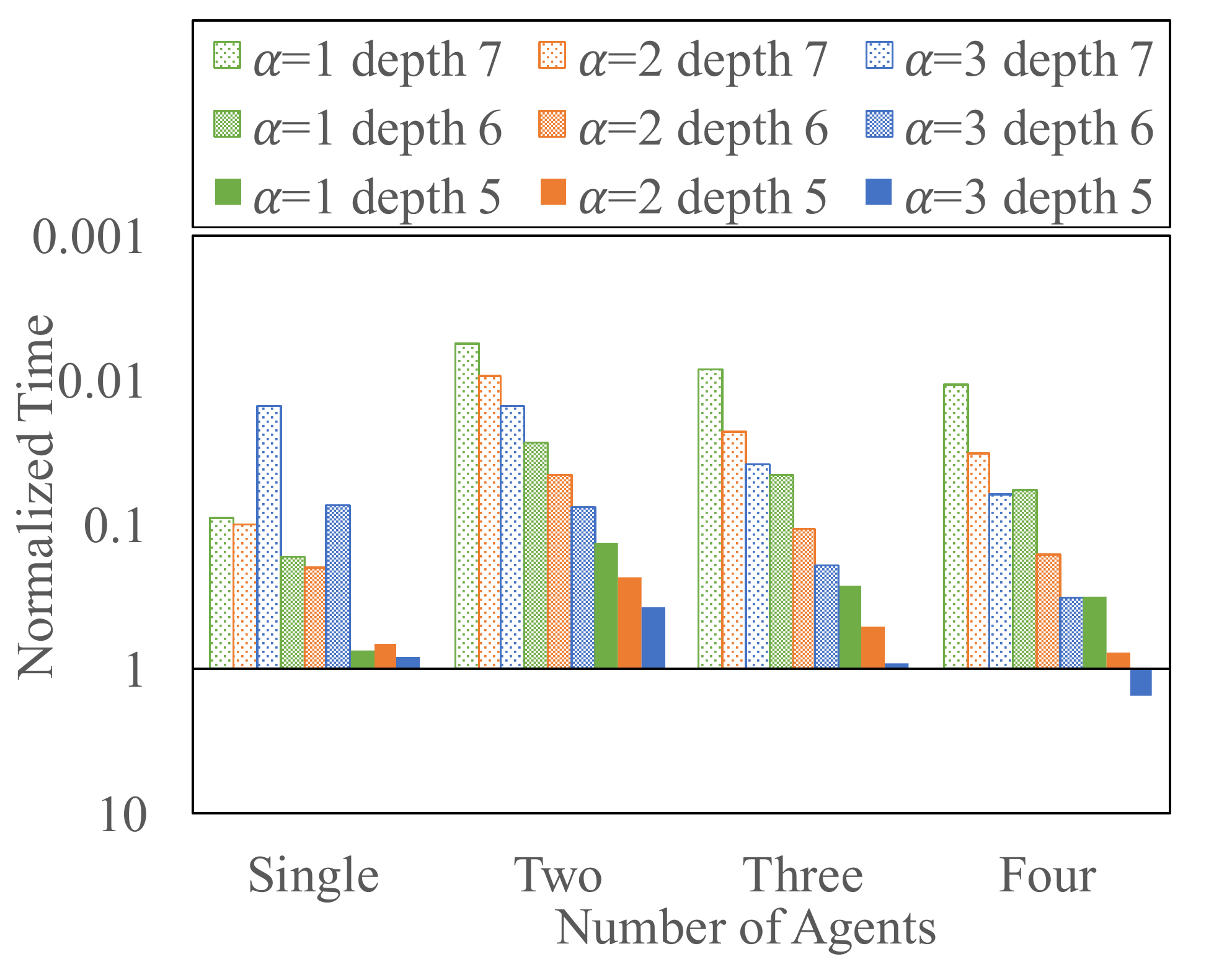}
		\caption{Computation time and solution cost of fine path for different number of agents and parameters normalized by regular A* result. 
		Left: normalized path cost. Right: normalized time.}
		\label{f:holistic_result}
	\end{figure}
	
	In the single agent case, the performance was highly sensitive to the choice of $\alpha$.
	At some abstraction structure, the single agent could still find near optimal solutions, but in the worst case, the path length was 16 times longer than the optimal. In contrast, as a better heuristic for refinement was used, the multi-agent cases resulted in more consistent computation times and cost performance that was
	less sensitive to the particular choice of $\alpha$.   
	
	\section{Conclusion}
	
	In this paper, we propose a new algorithm to solve a single query shortest path planning problem using multiple multi-resolution graphs representing the same search space.
	The solution quality and the speeds up from abstraction 
	is balanced efficiently by using multiple agents distributed in the search space, as they communicate only the expansion result of A* to avoid unnecessary communication. 
	The completeness and  optimality of the algorithm are shown. 
	The proposed scheme was applied to a backtracking algorithm and demonstrated the advantages of using selective and distributed information provided by other agents for refinement heuristics in terms of computational time and solution quality.

	\section*{Acknowledgement}
	
	This work has been supported by ARL under DCIST CRA W911NF-17-2-0181.
	
	\bibliographystyle{bib/IEEEtran}
	\bibliography{bib/mams}
	
\end{document}